\documentclass[twocolumn,aps]{revtex4}
\usepackage{mathrsfs}
\usepackage{amsfonts}
\usepackage{dcolumn}% Align table columns on decimal point
\usepackage{bm}% bold math
\usepackage{epsfig}
\usepackage{epstopdf} % 自动转换 EPS 为 PDF
\usepackage{pgf,fancyhdr}
\usepackage{braket}%为使用狄拉克符号而调用。
\usepackage{amsmath}%数学公式宏包
\usepackage{amssymb}%更多数学符号宏包，我们主要为使用数集的符号而调用。
\usepackage{amsthm}%定理环境包
\usepackage{graphicx}%使用插图宏包。
\usepackage[colorlinks=true, linkcolor=blue, citecolor=blue, urlcolor=red]{hyperref}
\usepackage{calrsfs}
\DeclareMathAlphabet{\pazocal}{OMS}{zplm}{m}{n}% 建立pazocal数学字体命令
\newtheorem{theorem}{Theorem} % 按章节编号 (Theorem 1.1, 1.2...)
\newtheorem{corollary}{Corollary}
\theoremstyle{definition} % 标题粗体，内容正体

\newtheorem{example}{Example}
\theoremstyle{remark}    % 标题斜体，内容正体

\begin{document}
\title{Mutually unbiased measurements-induced lower bounds of concurrence}
\author{Yu Lu$^{1}$}%
%\email[]{18510598788@163.com}
\author{Meng Su$^{1}$}%
%\email[]{nora24680679@163.com}
\author{Zhong-Xi Shen$^{2}$}
%\email[]{18738951378@163.com}
\author{Hong-Xing Wu$^{1, 3}$}
%\email[]{jxsruwhx@163.com$^{1}$,$^{1}$}
\author{Shao-Ming Fei$^{1}$}
%\email[]{feishm@cnu.edu.cn}
\author{Zhi-Xi Wang$^{1}$}
%\email[]{wangzhx@cnu.edu.cn}
\affiliation{%
$^{1}$School of Mathematical Sciences, Capital Normal University, Beijing 100048, China\\
$^{2}$School of Mathematics and Science, Nanyang Institute of Technology, Nanyang, Henan 473004, China\\
$^{3}$School of Mathematics and Computational Science, Shangrao Normal University, Shangrao 334001, China}

\begin{abstract}
We propose a family of lower bounds for concurrence in quantum systems using mutually unbiased measurements, which prove more effective in entanglement estimation compared to existing methods. Through analytical and numerical examples, we demonstrate that these bounds outperform conventional approaches, particularly in capturing finer entanglement features. Additionally, we introduce separability criterions based on MUMs for arbitrary $d$-dimensional bipartite systems, the research results show that our criterion has more advantages than the existing criteria.

\end{abstract}


\maketitle

\section{Introduction}
Quantum entanglement, as one of the most profound features of quantum mechanics,  plays a pivotal role in quantum information 
processing tasks such as quantum cryptography, teleportation, and computing \cite{Horodecki2009}. An important issue in the 
theory of quantum entanglement is the quantification and estimation of entanglement for composite systems.  
The concurrence \cite{rungta2001universal} is one of the most widely used entanglement measures in quantum information theory. 
However, due to the infimum involved in its calculation, analytical expressions for concurrence are only available for certain 
specific quantum states. This limitation has motivated extensive research efforts to establish reliable lower bounds of 
concurrence for general bipartite quantum states 
\cite{peres1996separability,horodecki1996necessary,chen2002matrix,rudolph2005further,chen2005concurrence}. 
Recently, symmetric measurements, including Mutually unbiased measurements (MUMs) and general symmetric informationally 
complete positive operator valued measurements (GSIC -POVMs),  presenting tighter bounds by exploiting their conical 2-design 
properties \cite{wang2025symmetric}.%

Another important problem proposed is how to distinguish quantum entangled states from the states without entanglement, i.e., separable states.
Over time, various separability criteria have emerged, such as the positive partial transposition criterion \cite{p96,ho96,ho97}, the realignment criterion \cite{r1,chen2002matrix,r3,r4,r5}, the covariance matrix criterion \cite{c1}, and the correlation matrix criterion \cite{c2,c3}.

MUMs have emerged as powerful tools for entanglement detection.

 A complete set of \(d+1\) MUMs can always be constructed for arbitrary dimensions \cite{ka14}. This advantage enables MUMs-based criteria to outperform mutually unbiased bases (MUBs) methods in detecting entanglement, particularly for isotropic states \cite{Chen2014}. Concurrently, symmetric informationally complete positive operator-valued measurements (SIC-POVMs) and their generalized versions offer alternative frameworks for entanglement detection \cite{Gour2014}. These symmetric measurements not only simplify criteria but also enhance experimental accessibility \cite{Shang2018}.

In this paper, based on MUMs, we derive the lower bound of concurrence, and deeply explore the separability problem by means of this measurement method. In addition, a separability criterion applicable to two-qudit systems (two-qubit systems) is proposed.

\section{Preliminary}

In quantum information theory, MUBs serve as fundamental mathematical tools. Let
$\mathcal{B}_{1}=\{|b_{i}\rangle\}_{i=1}^{d}$ and
$\mathcal{B}_{2}=\{|c_{j}\rangle\}_{j=1}^{d}$ be two orthonormal bases in $\mathbb{C}^{d}.$
They are called mutually unbiased if they satisfy:
$$
|\langle b_{i}|c_{j}\rangle|=\frac{1}{\sqrt{d}} ,~~~\forall\, i,j=1,2,\cdots,d.
$$
A set of orthonormal bases $\{\mathcal{B}_{1},
\mathcal{B}_{2},\cdots,\mathcal{B}_{m}\}$ in $\mathbb{C}^{d}$ is characterized as a set of mutually unbiased bases if any two distinct bases within this collection are mutually unbiased. Since each $\mathcal{B}_{k}$ can be
written as $d$ rank one projectors summing to  the identity
operator, these MUBs correspond to $m$ projective measurements, formulated as positive operator-valued measurements (POVMs), on a quantum system of dimension $d$. Consequently, if a physical system is prepared in
an eigenstate of basis $\mathcal{B}_{k}$ and measured in basis
$\mathcal{B}_{k'}$, then all the measurement outcomes are equally probable.

In Ref.\cite{ka14}, the authors introduced the concept of MUMs. Two POVMs on $\mathbb{C}^{d}$, denoted by
$\pazocal{P}^{(b)}=\{P_{n}^{(b)}\}_{n=1}^{d}$, $b=1,2$, are said to be
mutually unbiased measurements if
\begin{equation}\label{eq1}
\begin{split}
\mathrm{Tr}(P_{n}^{(b)})&=1,\\
\mathrm{Tr}(P_{n}^{(b)}P_{n'}^{(b')})&=\frac{1}{d},~~~b\neq b',\\
\mathrm{Tr}(P_{n}^{(b)}P_{n'}^{(b')})&=\delta_{n,n'}\,\kappa+(1-\delta_{n,n'})\frac{1-\kappa}{d-1},
\end{split}
\end{equation}
where $\frac{1}{d} < \kappa \leq 1$, and $\kappa=1$ if and only if all
$P_{n}^{(b)}$'s are rank one projectors, i.e., $\pazocal{P}^{(1)}$
and $\pazocal{P}^{(2)}$ correspond to MUBs.

A general construction of $d+1$ MUMs was introduced in Ref.\cite{ka14}.
Consider a collection of $d^2 - 1$ Hermitian, traceless operators $\{F_{n,b}:n=1,2,\cdots,d-1, b=1,2,\cdots,d+1\}$ acting on $\mathbb{C}^d $, which satisfy the orthogonality condition
 $\mathrm{\rm Tr}(F_{n,b}F_{n',b'})=\delta_{n,n'}\delta_{b,b'}$.
Define $d(d+1)$ operators
\begin{equation}\label{eq2}
F_{n}^{(b)}=
\begin{cases}
   F^{(b)}-(d+\sqrt{d})F_{n,b},&n=1,2,\cdots,d-1;\\[2mm]
   (1+\sqrt{d})F^{(b)},&n=d,
\end{cases}
\end{equation}
where $F^{(b)}=\sum_{n=1}^{d-1}F_{n,b}$, $b=1,2,\cdots,d+1.$  Then $d+1$ MUMs can be explicitly constructed \cite{ka14}:
\begin{equation}\label{eq3}
P_{n}^{(b)}=\frac{1}{d}I+tF_{n}^{(b)},
\end{equation}
with $b=1,2,\cdots,d+1, n=1,2,\cdots,d,$ and the value of $t$ must be selected to ensure that $P_{n}^{(b)} \geq 0$. The condition specifies that the parameter
$t$ must satisfy the inequality
\begin{equation}\label{tIneq}
-\frac1{d}\frac1{\lambda_{\rm max}}\leq t\leq\frac1{d}\frac1{|\lambda_{\rm min}|}
\end{equation}
where $\lambda_{\rm min} = \min_b\lambda_{\rm min}^{(b)}$, $\lambda_{\rm max} = \max_b\lambda_{\rm max}^{(b)}$, and $\lambda_{\rm min}^{(b)}$ and $\lambda_{\rm max}^{(b)}$ the largest positive and smallest negative eigenvalues, respectively,  across all operators $F_n^{(b)}$ from Eq.(\ref{eq2}) for $n = 1,\ldots,d$.

 %Any set of $d+1$ MUMs can be formulated in this form.

Associated with the construction of MUMs (\ref{eq3}), the parameter $\kappa$ is given by
\begin{equation}\label{Ek}
\kappa=\frac{1}{d}+t^{2}(1+\sqrt{d})^{2}(d-1).
\end{equation}
If the operators $F_{n,b}$ are taken to be the generalized Gell-Mann operator basis,
then $\kappa=\frac{1}{d} + \frac{2}{d^{2}}$
is the optimally \cite{ka14}. That is to say, for any $d$, there always exist $d+1$ MUMs, which contrasts with MUBs, where this is only possible for prime power dimensions $d$. However, if the parameter $\kappa$ is fixed, the two values of $t$ may not ensure  the
existence of operator basis such that $P_{n}^{(b)}$'s are positive.
For instance, if $\kappa=1$ and there exist $d+1$ MUMs, then
it must form a complete set of MUBs.  Yet, it remains unknown whether
$d+1$ MUBs exist when $d$ is not a prime power.

%%%%%%%%%%%%%%%%%%%%%%%%%%%%%%%%%%%%%%%%%%%%%%%%%%%%%%%%%%%%%%%%%%%%%%%%%%%%%%%%%%%%%%%%%%%%%%%
%%%%%%%%%%%%%%%%%%%%%%%%%%%%%%%%%%%%%%%%%%%%%%%%%%%%%%%%%%%%%%%%%%%%%%%%%%%%%%%%%%%%%%%%%%%%%%%

\section{Mutually unbiased measurements based lower bounds of concurrence}

Let $\mathcal{H}_{A}$ and $\mathcal{H}_{B}$ be $d$-dimensional complex Hilbert space. The concurrence of a pure state $\ket{\psi} \in \mathcal{H}_{A}\otimes \mathcal{H}_{B}$ is defined by
\begin{equation*}
	C(\ket{\psi}) = \sqrt{2(1-{\rm Tr}(\rho_{A}^{2}))} = \sqrt{2(1-{\rm Tr}(\rho_{B}^{2}))},
\end{equation*}
where $\rho_{A(B)} = {\rm Tr}_{B(A)}(\ket{\psi}\bra{\psi})$ denotes the reduced density state obtained by tracing out the second subsystem. The concurrence is extended to mixed states $\rho$ through the convex roof extension,
\begin{equation*} C(\rho)=\min\limits_{\{p_{i},\ket{\psi_{i}}\}}\sum\limits_{i}p_{i}C(\ket{\psi_{i}}),
\end{equation*}
where the minimum is determined all possible pure state decompositions of $\rho = \sum\limits_{i}p_{i}\ket{\psi_{i}}\bra{\psi_{i}}$, with the conditions $p_{i}\geq 0$ and $\sum\limits_{i}p_{i}=1$.

\begin{theorem}
Let $\rho$ is a density matrix in $\mathbb{C}^{d}\bigotimes\mathbb{C}^{d}$,
consider two sets of measurements, $\{\pazocal{P}^{(b)}\}_{b = 1}^{d + 1}$ and $\{\pazocal{P}^{(b')}\}_{b' = 1}^{d + 1}$, each consisting of $d + 1$ MUMs on $\mathbb{C}^{d}$, two sets of measurements, sharing the same parameter $\kappa$. Specifically, $\pazocal{P}^{(b)} = \{P_{n}^{(b)}\}_{n = 1}^{d}$ and $\pazocal{P}^{(b')} = \{P_{n'}^{(b')}\}_{n' = 1}^{d}$, for $b,b' = 1,2,\cdots,d + 1$, let $\{\ket{\omega_{n,b}}|n = 1,\cdots,d; b = 1,\cdots,d + 1\}$ be an orthonormal basis of the vector space $\mathbb{C}^{d(d + 1)}$. We then define the operator $\mathcal{J}(\rho)$ as follows:

\[
\mathcal{J}(\rho)=\sum_{n,n' = 1}^{d}\sum_{b,b' = 1}^{d + 1}\text{\rm Tr}\left(\rho\left(P_{n}^{(b)}\otimes P_{n'}^{(b')}\right)\right)\ket{\omega_{n,b}}\bra{\omega_{n',b'}}.
\]

The concurrence $C(\rho)$ satisfies
	\begin{equation}\label{eq4}
		C(\rho) \geq \sqrt{\frac{2(d-1)}{d(\kappa d-1)}}(\|\mathcal{J}(\rho)\|_{\rm Tr} - (1+\kappa)).
	\end{equation}
\end{theorem}

\begin{proof}
Consider the optimal pure state decomposition
 $\{p_i,\ket{\psi_i}\}$ of $\rho$ that satisfies $C(\rho) = \sum\limits_{i}p_{i}C(\ket{\psi_i})$.  Because of the convex property of the trace norm, which ensures that $\sum\limits_{i}p_{i}\|\mathcal{J}(\ket{\psi_{i}}\bra{\psi_{i}})\|_{\rm Tr}\geq \|\mathcal{J}(\rho)\|_{\rm Tr},$  it suffices to prove the theorem for pure states $\ket{\psi} \in \mathcal{H}_{A}\otimes \mathcal{H}_{B}$,
\begin{equation}\label{eq5}
C(\ket{\psi}) \geq \sqrt{\frac{2(d-1)}{d(\kappa d-1)}}(\|\mathcal{J}{\ket{\psi}\bra{\psi}}\|_{\rm Tr} - (1+\kappa)).
\end{equation}
Using the Schmidt decomposition, we can find orthonormal bases  $\ket{\varphi_1},\cdots,\ket{\varphi_{d}}$ and $\ket{\phi_1},\cdots,\ket{\phi_{d}}$ in $\mathcal{H}_{A}$ and $\mathcal{H}_{B}$ such that $\ket{\psi} = \sum\limits_{i = 1}^{r}\lambda_{i}\ket{\varphi_i}\otimes\ket{\phi_{i}}$, where $\lambda_{i}\geq 0$ and $\sum\limits_{i = 1}^{r}\lambda_{i}^{2} = 1$. Denote $\ket{\mu_{ij}}=\sum\limits_{n=1}^{d}\sum\limits_{b=1}^{d+1}
\braket{\varphi_{j}|F_{n}^{(b)}|\varphi_{i}}\ket{\omega_{n,b}}$ and $\ket{\nu_{ij}}=\sum\limits_{n=1}^{d}
\sum\limits_{b=1}^{d+1}\braket{\phi_{j}|F_{n}^{(b)}|\phi_{i}}\ket{w_{\omega,b}}$. Subsequently,

    \begin{equation}\label{eq6}
        \begin{split}
            &\mathcal{J}(|\psi\rangle\langle\psi|) \\&= \sum_{n,n' = 1}^{d} \sum_{b,b' = 1}^{d + 1} \sum_{i,j = 1}^{r} \lambda_{i} \lambda_{j} \langle \varphi_{j}|F_{n'}^{(b')}|\varphi_{i}\rangle \langle \phi_{j}|F_{n}^{(b)}|\phi_{i}\rangle |\psi_{n,b}\rangle\\
            &= \sum_{i = 1}^{r} \sum_{j = 1}^{r} \lambda_{i} \lambda_{j} |\mu_{ij}\rangle \langle\nu_{ij}|
        \end{split}
    \end{equation}
where $\ket{\overline{\nu}_{ij}}=\sum\limits_{n'=1}^{d}
\sum\limits_{n'=1}^{d+1}\overline{\braket{\phi_{j}|F_{n'}^{(b')}|\phi_{i}}}\ket{\omega_{n',b'}}$ with
$\overline{\braket{\phi_{j}|F_{n'}^{(b')}|\phi_{i}}}$ representing the complex conjugation of ${\braket{\phi_{j}|F_{n'}^{(b')}|\phi_{i}}}$.\par

(1) Firstly, we analyze the situation where $\ket{\mu_{ij}}=\ket{\overline{\nu}_{ij}}$ for every $i,j \in \{1,\cdots,d\}$.  Under these conditions, it is clear that $\mathcal{J}(\ket{\psi}\bra{\psi})$ is a positive semidefinite, and then
    \begin{equation}\label{eq7}
    	\|\mathcal{J}(\ket{\psi}\bra{\psi})\|_{\rm Tr}={\rm Tr}\left(\mathcal{J}(\ket{\psi}\bra{\psi})\right)
    =\sum\limits_{i=1}^{r}\sum\limits_{j=1}^{r}\lambda_{i}\lambda_{j}\braket{\mu_{ij}|\mu_{ij}}.
    \end{equation}
Define $\mathbb{F}=\sum\limits_{i=1}^{d}\sum\limits_{j=1}^{d}\ket{\varphi_{i}}\bra{\varphi_{j}}\otimes\ket{\varphi_{j}}\bra{\varphi_{i}}$. Based on the relation $\sum\limits_{n = 1}^{d}\sum\limits_{b = 1}^{d+1}F_{n}^{(b)}\otimes F_{n}^{(b)}= (1+ \dfrac{1-\kappa}{d - 1})\mathbb{I}+\dfrac{\kappa d-1}{(d-1)}\mathbb{F}$ \cite{wang2018uncertainty}, we obtain
    \begin{flalign*}
    	& \begin{array}{ll}
    		&\braket{\mu_{ij}|\mu_{i^{\prime}j^{\prime}}}
    		\hspace{-2mm} \\& = \sum\limits_{n,n' = 1}^{d}\sum\limits_{b,b' = 1}^{d+1}
    \overline{\braket{\varphi_{j}|F_{n}^{(b)}|\varphi_{i}}}\braket{\varphi_{j^{\prime}}
    |F_{n'}^{(b')}|\varphi_{i^{\prime}}}\braket{\omega_{n,b}|w_{n',b'}}\\
    &
    = \sum\limits_{n = 1}^{d}\sum\limits_{b = 1}^{d+1}\braket{\varphi_{i}|F_{n}^{(b)}
    |\varphi_{j}}\braket{\varphi_{j^{\prime}}|F_{n}^{(b)}|\varphi_{i^{\prime}}}\\
    		& = \sum\limits_{n = 1}^{d}\sum\limits_{b = 1}^{d+1}\braket{\varphi_{i}\varphi_{j^{\prime}}
    |F_{n}^{(b)}\otimes F_{n}^{(b)}|\varphi_{j}\varphi_{i^{\prime}}}\\
    &
    = \braket{\varphi_{i}\varphi_{j^{\prime}}|\sum\limits_{n = 1}^{d}\sum\limits_{b = 1}^{d+1}F_{n}^{(b)}\otimes F_{n}^{(b)}|\varphi_{j}\varphi_{i^{\prime}}}\\
    		& = \dfrac{\kappa d-1}{(d-1)}\braket{\varphi_{i}\varphi_{j^{\prime}}|\mathbb{F}|\varphi_{j}\varphi_{i^{\prime}}}+(1+ \dfrac{1-\kappa}{d - 1})\braket{\varphi_{i}\varphi_{j^{\prime}}|\mathbb{I}|\varphi_{j}\varphi_{i^{\prime}}}\\
    		& = \dfrac{\kappa d-1}{(d-1)}\delta_{ii^{\prime}}\delta_{jj^{\prime}}
    +(1 + \dfrac{1-\kappa}{d - 1})\delta_{ij}\delta_{i^{\prime}j^{\prime}}.
    	\end{array}&
    \end{flalign*}
    Therefore, $\braket{\mu_{ij}|\mu_{ij}} = \dfrac{\kappa d-1}{(d-1)}
    +(1 + \dfrac{1-\kappa}{d - 1})\delta_{ij}$ and

\begin{align}\label{eq8}
    		&\|\mathcal{J}(\ket{\psi}\bra{\psi})\|_{\rm Tr}
    		\hspace{-2mm} \nonumber\\&=\dfrac{\kappa d-1}{(d-1)}\sum\limits_{i=1}^{r}\sum\limits_{j=1}^{r}\lambda_{i}\lambda_{j}+(1+ \dfrac{1-\kappa}{d - 1})\sum\limits_{i=1}^{r}\sum\limits_{j=1}^{r}\delta_{ij}\lambda_{i}\lambda_{j}\nonumber\\
    		& =\dfrac{\kappa d-1}{(d-1)}\sum\limits_{i=1}^{r}\sum\limits_{j=1}^{r}\lambda_{i}\lambda_{j}+(1+ \dfrac{1-\kappa}{d - 1})\sum\limits_{i=1}^{r}\lambda_{i}^{2}\nonumber\\
    		& =\dfrac{\kappa d-1}{(d-1)}\sum\limits_{i\neq j}\lambda_{i}\lambda_{j}+(1+\kappa)\sum\limits_{i=1}^{r}\lambda_{i}^{2}\nonumber\\
    		& =\dfrac{2(\kappa d-1)}{(d-1)}
    \sum\limits_{i<j}\lambda_{i}\lambda_{j}+1+\kappa.
    \end{align}

(2)  We now proceed to prove the general case by leveraging the previously established results. Observe that $\braket{\overline{\nu}_{ij}|\overline{\nu}_{i^{\prime}j^{\prime}}}
=\overline{\braket{\nu_{ij}|\nu_{i^{\prime}j^{\prime}}}}$ and $\braket{\nu_{ij}|\nu_{i^{\prime}j^{\prime}}}=\dfrac{\kappa d-1}{(d-1)}
\delta_{ii^{\prime}}\delta_{jj^{\prime}}+(1+ \dfrac{1-\kappa}{d - 1})
\delta_{ij}\delta_{i^{\prime}j^{\prime}}\in\mathbb{R}$. Therefore,  $\braket{\overline{\nu}_{ij}|\overline{\nu}_{i^{\prime}j^{\prime}}}=\braket{\nu_{ij}
|\nu_{i^{\prime}j^{\prime}}}=\braket{\mu_{ij}|\mu_{i^{\prime}j^{\prime}}}$. This implies the existence of a unitary operator $U$ on $\mathcal{H}_{A}\otimes \mathcal{H}_{B}$ such that $\ket{\overline{\nu}_{ij}}=U\ket{\mu_{ij}}$ for $i,j=1,2,\cdots,d$.  Accordingly, we obtain

\begin{align}\label{eq9}
\|\mathcal{J}(|\psi\rangle\langle\psi|)\|_{\mathrm{\rm Tr}} &= \left\|\sum_{i = 1}^{r} \sum_{j = 1}^{r} \lambda_{i} \lambda_{j}|\mu_{i j}\rangle\langle\overline{\nu}_{ij}|\right\|_{\mathrm{\rm Tr}} \nonumber\\
&= \left\|\left(\sum_{i = 1}^{r} \sum_{j = 1}^{r} \lambda_{i} \lambda_{j}|\mu_{i j}\rangle\langle\mu_{i j}|\right) U^{\dagger}\right\|_{\mathrm{\rm Tr}} \nonumber\\
&= \left\|\sum_{i = 1}^{r} \sum_{j = 1}^{r} \lambda_{i} \lambda_{j}|\mu_{i j}\rangle\langle\mu_{i j}|\right\|_{\mathrm{\rm Tr}} \nonumber
\\
&= \frac{2(\kappa d - 1)}{(d - 1)} \sum_{i<j} \lambda_{i} \lambda_{j}+1 + \kappa.
\end{align}
where the last equality is derived from (\ref{eq7}) and (\ref{eq8}). Given that $(C(\ket{\psi}))^{2}\geq \dfrac{8}{d(d-1)}\left(\sum\limits_{i<j}\lambda_{i}\lambda_{j}\right)^{2}$ \cite{Chen2005}, it follows that
$C(\ket{\psi})\geq 2\sqrt{\dfrac{2}{d(d-1)}}\sum\limits_{i<j}\lambda_{i}\lambda_{j}$,  This result, combined with (\ref{eq9}), leads to (\ref{eq5}).
\end{proof}

%%%%%%%%%%%%%%%%%%%%%%%%%
We give several examples below to illustrate the advantages of our results.

\begin{example}
We consider a state by mixing $\rho$ with white noise,
\begin{equation*}
    \rho_{p}=\dfrac{1-p}{9}I_{9}+p\rho.
\end{equation*}
	
    where
    $\rho$ is a $3\otimes 3$ PPT entangled state:
	\begin{equation*}
		\rho = \dfrac{1}{4}\left(I-\sum\limits_{i=0}^{4}\ket{\psi_{i}}\bra{\psi_{i}}\right),
	\end{equation*}
    \begin{align*}
    &\ket{\psi_{0}}=\dfrac{\ket{0}(\ket{0}-\ket{1})}{\sqrt{2}},\\ &\ket{\psi_{1}}=\dfrac{(\ket{0}-\ket{1})\ket{2}}{\sqrt{2}},\\ &\ket{\psi_{2}}=\dfrac{\ket{2}(\ket{1}-\ket{2})}{\sqrt{2}},\\  &\ket{\psi_{3}}=\dfrac{(\ket{1}-\ket{2})\ket{0}}{\sqrt{2}},\\ &\ket{\psi_{4}}=\dfrac{(\ket{0}+\ket{1}+\ket{3})(\ket{0}+\ket{1}+\ket{3})}{3}.
    \end{align*}

%Setting $t=0.12,$  we have $C(\rho)\geq 0.044183$ from our Theorem.
The complete sets of MUMs are constructed from the generalized Gell-Mann operators. Regarding the detailed construction process, one can refer to Ref.\cite{ka14}  the parameter $\kappa$ is given by $\kappa=\frac{1}{d}+t^{2}(1+\sqrt{d})^{2}(d-1),$ with $t\in[-0.10939, 0.122008]$

In Fig.\ref{fig:1}, we set $t$ as a variable and conducted a multi - dimensional comparative analysis with the results in \cite{shi2023family} and \cite{wang2025symmetric}. From the comparison data, it is clearly evident that our results are superior.
In Fig.\ref{fig:2}, when validating the results, we set the parameter $t=0. 12$ and compared it with the results presented in Ref.\cite{shi2023family} and \cite{wang2025symmetric}. Through this comparison, we clearly found that the results obtained from this study are superior to those in the references.
\begin{figure}[htp]
\centering
\includegraphics[width=0.50\textwidth]{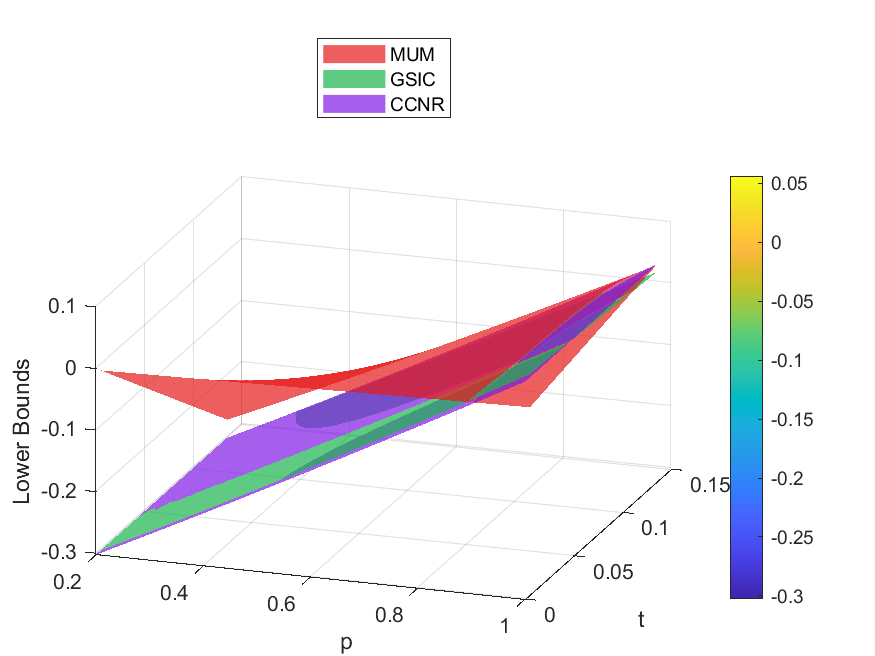}
\caption{$f_1$ from our result (red), $f_2$ from the Theorem $1$ in \cite{wang2025symmetric} (green) and $f_3
(x)$ from the Theorem 3 in \cite{shi2023family} (purple).}
\label{fig:1}
\end{figure}

\begin{figure}[htp]
\centering
\includegraphics[width=0.50\textwidth]{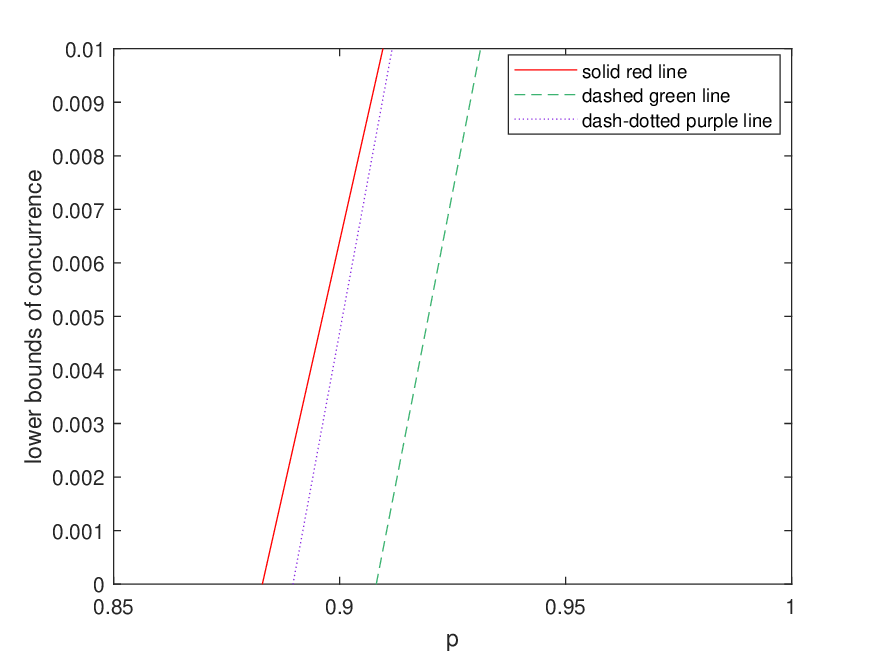}
\caption{$f_1$ from our result (solid red line), $f_2$ from the Theorem $1$ in \cite{wang2025symmetric} (dashed green line) and $f_3
(x)$ from the Theorem 3 in \cite{shi2023family} (dash-dotted purple line).}
\label{fig:2}
\end{figure}
    \end{example}

\begin{example}\label{Ex1}
	Consider the mixture of the bound entangled state proposed by Horodecki \cite{horodecki1997separability},
	\begin{equation*}
		\rho_{\upsilon}=\dfrac{1}{1+8\upsilon}
		\begin{pmatrix}
			\upsilon & 0 & 0 & 0 & \upsilon & 0 & 0 & 0 & \upsilon\\
			0 & \upsilon & 0 & 0 & 0 & 0 & 0 & 0 & 0\\
			0 & 0 & \upsilon & 0 & 0 & 0 & 0 & 0 & 0\\
			0 & 0 & 0 & \upsilon & 0 & 0 & 0 & 0 & 0\\
			\upsilon & 0 & 0 & 0 & \upsilon & 0 & 0 & 0 & \upsilon\\
			0 & 0 & 0 & 0 & 0 & \upsilon & 0 & 0 & 0\\
			0 & 0 & 0 & 0 & 0 & 0 & \frac{1+\upsilon}{2} & 0 & \frac{\sqrt{1-\upsilon^{2}}}{2}\\
			0 & 0 & 0 & 0 & 0 & 0 & 0 & \upsilon & 0\\
			\upsilon & 0 & 0 & 0 & \upsilon & 0 & \frac{\sqrt{1-\upsilon^{2}}}{2} & 0 & \frac{1+\upsilon}{2}
		\end{pmatrix}
	\end{equation*}
    and the $9\times 9$ identity matrix $I_{9}$,
    \begin{equation*}
    	\rho(\upsilon,q)=q\rho_{\upsilon}+\frac{1-q}{9}I_{9}.
    \end{equation*}
    Fig.\ref{fig:3} illustrate the lower bounds of $C(\rho(\upsilon,q))$ for $q=0.995$, we set $t$ as a variable and conducted a multi - dimensional comparative analysis with the results in \cite{lu2024separability} and \cite{wang2025symmetric}. From the comparison data, it is clearly evident that our results are superior.
In Fig.\ref{fig:4}, when validating the results, we set the parameter $t=0.08$ and compared it with the results presented in Ref.\cite{lu2024separability} and \cite{wang2025symmetric}. Through this comparison, we clearly found that the results obtained from this study are superior to those in the references.

\begin{figure}[htp]
\centering
\includegraphics[width=0.50\textwidth]{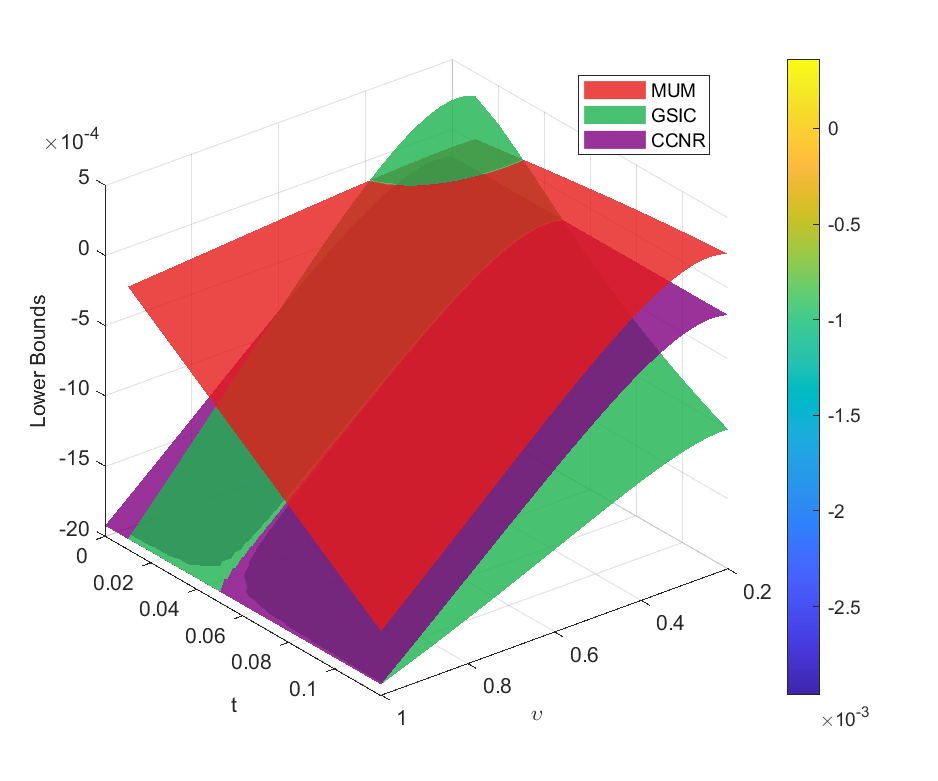}
\caption{$g_1$ from our result (red), $f_2$ from the Theorem $1$ in \cite{wang2025symmetric} (green) and $g_3
(x)$ from the Theorem 3 in \cite{shi2023family} (purple).}
\label{fig:3}
\end{figure}

 \begin{figure}[htp]
\centering
\includegraphics[width=0.50\textwidth]{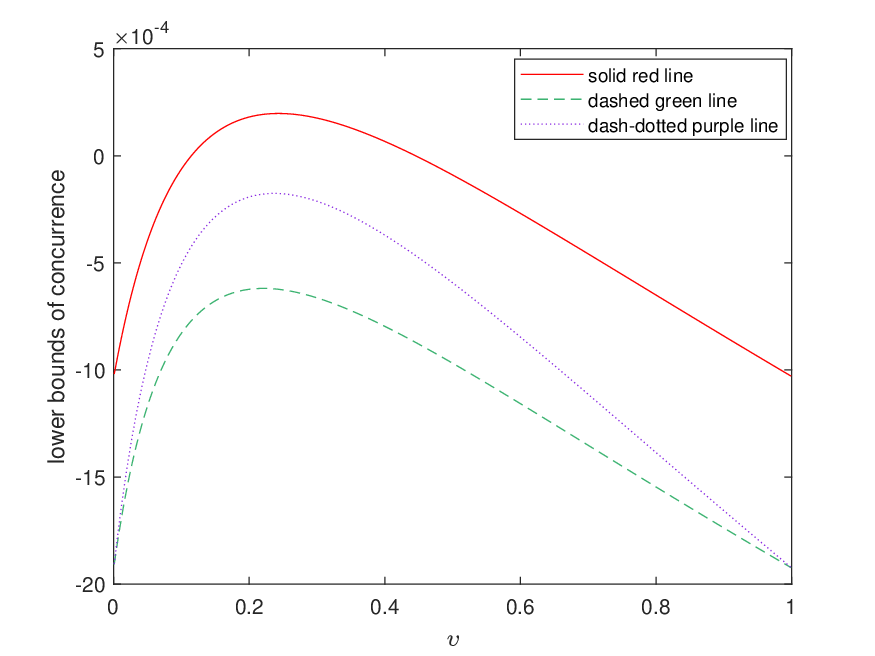}
\caption{$g_1$ from our result (solid red line),  $g_2
(x)$ from the Corollary $1$ in \cite{wang2025symmetric} (dashed green line) and $g_3
(x)$ from the Theorem $2$ in \cite{lu2024separability} (dash-dotted purple line).}
\label{fig:4}
\end{figure}

\end{example}
%%%%%%%%%%%%%%%%%%%%%%%%%%%%%%%%%%%%%%%%%%%%%%%%%%%%%%%%%%%%%%%%%%%%%%%%%%%%%%%%%%%%%%%%%%%%%%%
%%%%%%%%%%%%%%%%%%%%%%%%%%%%%%%%%%%%%%%%%%%%%%%%%%%%%%%%%%%%%%%%%%%%%%%%%%%%%%%%%%%%%%%%%%%%%%%
\section{Entanglement detection using MUMs}

Bipartite entanglement is a core concept in quantum information. It describes that the quantum state of a composite quantum system $\ket{\psi}\in \mathcal{H}_{A}\otimes \mathcal{H}_{B}$ cannot be represented as the direct product of the individual quantum states of the two subsystems. This entangled state exhibits a high degree of correlation between the two systems, such that a measurement on one system will instantaneously affect the state of the other system, even if they are separated by a large distance. In the following, we will present a criterion for bipartite entanglement from the perspective of the Schmidt rank.

\begin{theorem}\label{Th1}
For a bipartite pure state that we write in its Schmidt decomposition $\ket{\psi}=\sum\limits_{i=1}^{r}\lambda_{i}\ket{\varphi_i}\otimes\ket{\phi_{i}}$, we have
\begin{equation}
\|\mathcal{J}(\ket{\psi}\bra{\psi})\|_{\rm Tr} \leq 1 + \kappa,
\end{equation}
where the number $r$ is the Schmidt rank of the pure state, $\kappa$ is defined in Eq.(\ref{Ek}); it is the
rank of the reduced density matrix $\rho_A = \rm Tr_B(\ket{\psi}\bra{\psi})$.
\end{theorem}
\begin{proof}
If we consider from the general two-body pure state perspective, according to the Schmidt decomposition $$\ket{\psi}=\sum\limits_{i=1}^{r}\lambda_{i}\ket{\varphi_i}\otimes\ket{\phi_{i}}, \sum_{i =1}^{r}\lambda_{i}^{2}=1, \lambda_{i} \geq 0.$$

In Ref.\cite{terhal2000schmidt}, by using Lagrange multipliers to implement the constraint $\sum_{i} \lambda_{i}^2 = 1$ one can show that $\left[\sum_{i = 1}^{r} \lambda_{i}\right]^{2} \leq r$, so
\begin{equation}\label{eq10}
2\sum\limits_{i<j}\lambda_{i}\lambda_{j} = \left[\sum_{i = 1}^{r} \lambda_{i}\right]^{2} - \sum_{i} \lambda_{i}^2 \leq r-1.
\end{equation}
Useing (\ref{eq5}) and (\ref{eq10}), we have
\begin{align*}
\|\mathcal{J}(\ket{\psi}\bra{\psi})\|_{\rm Tr} &= \dfrac{2(\kappa d-1)}{(d-1)}
    \sum\limits_{i<j}\lambda_{i}\lambda_{j}+1+\kappa \\
    &\leq   \frac{2(\kappa d - 1)}{(d - 1)} (r-1)+1 + \kappa
\end{align*}

 According to the Schmidt decomposition, if the Schmidt rank
$r=1$, then the quantum state $|\psi\rangle_{AB}$ is separableand can be expressed in the form of a product state; if the Schmidt rank $r > 1$, then the quantum state $|\psi\rangle_{AB}$ is entangled.

 If the quantum state $|\psi\rangle_{AB}$ is separableand, we have $r=1$, $$\|\mathcal{J}(\ket{\psi}\bra{\psi})\|_{\rm Tr} \leq 1 + \kappa$$
\end{proof}

By applying the Schmidt criterion, we successfully derived the separability criterion based on MUMs. Using the same derivation ideas and methods, the separability criterion based on symmetric informationally complete (SIC) was also successfully derived. For the detailed construction of the SIC, please refer to the Ref.\cite{wang2025symmetric}.
\begin{corollary}
  Let $\{E_{\alpha,k}|\alpha=1,\cdots,N; \,k=1,\cdots,M\}$ be a informationally complete $(N, M)$-POVM with the free parameter $x$ on the $d$ dimensional Hilbert space $\pazocal{H}$, $\rho$ be a bipartite state in $\pazocal{H}\otimes\pazocal{H}$. Using (\ref{eq10}) we have
  \begin{equation}\label{eq13}
	\|\pazocal{P}(\rho)\|_{\rm Tr}\leq \dfrac{(d-1)(xM^{2}+d^{2})}{dM(M-1)}
  \end{equation}
  \begin{proof}
  From Ref.\cite{wang2025symmetric}, we know that
  \begin{flalign}\label{12}
    	& \begin{array}{ll}
    		&\|\pazocal{P}(\ket{\psi}\bra{\psi})\|_{\rm Tr}
    		\hspace{-2mm} \\
    		& =\dfrac{(d-1)(xM^{2}+d^{2})}{dM(M-1)}+\dfrac{2(xM^{2}-d)}{M(M-1)}
    \sum\limits_{i<j}\lambda_{i}\lambda_{j}\\
    & \leq \dfrac{(d-1)(xM^{2}+d^{2})}{dM(M-1)}+\dfrac{2(xM^{2}-d)}{M(M-1)}
    (r-1).
    	\end{array}
    \end{flalign}

    Based on Eq.(\ref{eq10}), the inequality holds.When $\rho$ satisfies the separability condition, it can be deduced that $r = 1$, which enables Eq.(\ref{eq13}) to hold.
  \end{proof}
\end{corollary}
According to Theorem \ref{Th1}, we can conclude that once Eq.(\ref{eq9}) is not satisfied, the quantum state is an entangled state.  Next, we will present an intuitive calculation process and result through Example 3.

\begin{example}\label{ex:2}
We will continue to use the quantum states involved in Example \ref{Ex1} to further carry out in-depth relevant analyses.

Table \ref{tab:1} presents a comparison of the results obtained from our Theorem \ref{Th1}, the realignment criterion as stated in Theorem 1 of Ref.\cite{shi2023family}, and the results from Ref.\cite{sun2025separability} for various values of $\upsilon$. A careful examination of Table \ref{tab:1} reveals that, in terms of detecting the entanglement of the state $\rho_{p}$, our Theorem \ref{Th1} outperforms the criteria presented in \cite{shi2023family} and \cite{sun2025separability}.

\end{example}

\begin{center}
	\begin{table*}[htp!]
\caption{Entanglement of the state $\rho_{p}$ in Example \ref{ex:2} for different values of $\upsilon$}
		\label{tab:1}
		\begin{tabular}{cccc}
			\hline\noalign{\smallskip}
		$\upsilon$ \quad& Theorem $1$ in \cite{shi2023family} \quad&Theorem $1$ in \cite{sun2025separability} %with $\alpha = \beta = 2, l=10$
\quad & Our Theorem \ref{Th1}(t=0.01)\\
		\noalign{\smallskip}\hline\noalign{\smallskip}
		0.2 \quad& $0.9943 \leq p \leq 1$ \quad & $0.99408  \leq p \leq 1$ \quad& $0.994054  \leq p \leq 1$ \\
		
		0.4 \quad& $0.9948  \leq p \leq 1$ \quad& $0.99463  \leq p \leq 1$ \quad& $0.99461  \leq p \leq 1$ \\
		
		0.6 \quad& $0.9964  \leq p \leq 1$ \quad& $0.99627  \leq p \leq 1$ \quad& $0.99626 \leq p \leq 1$ \\
		
		0.8 \quad& $0.9982  \leq p \leq 1$ \quad& $0.998127  \leq p \leq 1$ \quad& $0.998123  \leq p \leq 1$ \\
		
		0.9 \quad& $0.9991  \leq p \leq 1$ \quad& $0.99906892  \leq p \leq 1$ \quad& $0.999067  \leq p \leq 1$ \\
		\noalign{\smallskip}\hline
		\end{tabular}
	\end{table*}
\end{center}

%%%%%%%%%%%%%%%%%%%%%%%%%%%%%%%%%%%%%%%%%%%%%%%%%%%%%%%%%%%%%%%%%%%%%%%%%%%%%%%%%%%%%%
%%%%%%%%%%%%%%%%%%%%%%%%%%%%%%%%%%%%%%%%%%%%%%%%%%%%%%%%%%%%%%%%%%%%%%%%%%%%%%%%%%%%%%
\section{Conclusions and discussions}

Based on MUMs, we propose a novel family of lower bounds for concurrence in quantum entanglement research, with their superiority in entanglement estimation verified through both theoretical and numerical analyses. Compared to conventional methods, this approach can more precisely capture finer entanglement features. Furthermore, for arbitrary $d$-dimensional bipartite systems, we establish  a separability criterion based on MUMs. Both theoretical proofs and numerical experiments confirm that the new criterion outperforms existing methods in detection capability, exhibiting higher sensitivity in identifying weakly entangled states and complex high-dimensional systems. These results provide more efficient tools for the quantitative characterization and experimental detection of quantum entanglement, with potential future extensions to multipartite entanglement analysis and quantum information task optimization.

\bigskip
\noindent{\bf Acknowlegements}
This work is supported by the National Natural Science Foundation of China under Grant 12171044, and the specific research fund of the Innovation Platform for Academicians of Hainan Province.

\end{document}